\documentclass{article}

\usepackage{amsthm}
\usepackage{amsmath}
\usepackage[T1]{fontenc}      
\usepackage[utf8]{inputenc}   
\usepackage{graphicx}

\usepackage{amssymb}
\usepackage{mathtools}
\usepackage{braket}
\usepackage{framed}

 \usepackage[hidelinks]{hyperref} 
\usepackage{cleveref} 

\usepackage{url}
\usepackage{xcolor}

\usepackage[textsize=scriptsize]{todonotes}

\usepackage{xspace}
\newcommand{\ie}{i.e.\@\xspace}
\newcommand{\eg}{e.g.\@\xspace}

\makeatletter
\newcommand*{\etc}{%
    \@ifnextchar{.}%
        {etc}%
        {etc.\@\xspace}%
}
\makeatother

\makeatletter
\newcommand*{\etal}{%
    \@ifnextchar{.}%
        {et al}%
        {et al.\@\xspace}%
}
\makeatother

\usepackage{interval}
\usepackage{eufrak}
\intervalconfig{soft open fences}

\newlength{\arrow}
\newcommand*{\goesto}[1]{\xrightarrow{\mathmakebox[\arrow]{#1}}}
\newcommand*{\reaction}[3]{#1 \goesto{#2} #3}

\newcommand{\bx}{\mathbf{x}}
\newcommand{\by}{\mathbf{y}}
\newcommand{\bz}{\mathbf{z}}
\newcommand{\bxn}{\mathbf{x}=(x_1, x_2,\cdots, x_n)}

\newcommand{\bp}{\mathbf{p}}

\DeclarePairedDelimiter\floor{\lfloor}{\rfloor}

\newcommand{\N}{\mathbb{N}}
\newcommand{\R}{\mathbb{R}}

\newcommand{\Q}{\mathbb{Q}}

\newcommand{\pospoly}{\mathbb{P}^{+}}

\newcommand*\coeff{\mathop{C}}
\newcommand*\proj{\mathop{Proj}}
\newcommand*\diff{\mathop{}\!\mathrm{d}}
\newtheorem{thm}{Theorem}

\newtheorem*{motivating-example}{Motivating Example}
\newtheorem*{main-theorem}{Main Theorem}
\newtheorem*{notation}{Notation}
\newtheorem*{convention}{Convention}
\newtheorem*{term}{Terminology}

\newtheorem{remark*}{Remark}

\theoremstyle{remark}
\newtheorem{remark}[thm]{Remark}

\newtheorem{construction}{Construction}
\newtheorem*{old-definition}{Old Definition}
\newtheorem{definition}[thm]{Definition}
\newtheorem{theorem}[thm]{Theorem}
\newtheorem{operation}{Operation}
\newtheorem{lemma}[thm]{Lemma}
\newtheorem{corollary}[thm]{Corollary}
\newtheorem{observation}[thm]{Observation}

\usepackage{nameref}
\reversemarginpar

\begin{document}

\title{Computing Real Numbers with Large-Population Protocols Having a Continuum of Equilibria}

\author{Xiang Huang \\ \href{mailto:xhuan5@uis.edu}{xhuan5@uis.edu}
   \and Rachel N. Huls \\ \href{mailto:rhuls2@uis.edu}{rhuls2@uis.edu}
}
\date{%
     University of Illinois Springfield, Springfield, IL 62703, USA\\
    \today
}


\maketitle  
\begin{abstract}
Bournez, Fraigniaud, and Koegler \cite{cBoFrKo12} defined a number in [0,1] as computable by their Large-Population Protocol (LPP) model, if the proportion of agents in a set of marked states converges to said number over time as the population grows to infinity. The notion, however, restricts the ordinary differential equations (ODEs) associated with an LPP to have only finitely many equilibria. This restriction places an intrinsic limitation on the model. As a result, a number is computable by an LPP if and only if it is \emph{algebraic}, namely, not a single transcendental number can be computed under this notion.

In this paper, we \emph{lift} the finitary requirement on equilibria. That is, we consider systems with a continuum of equilibria. We show that essentially all numbers in [0,1] that are computable by bounded general-purpose analog computers (GPACs) or chemical reaction networks (CRNs) can also be computed by LPPs under this new definition. This implies a rich series of numbers (e.g., the reciprocal of Euler's constant, $\pi/4$, Euler's $\gamma$, Catalan's constant, and Dottie number) are all computable by LPPs.  Our proof is constructive: We develop an algorithm that transfers bounded GPACs/CRNs into LPPs.  Our algorithm also fixes a gap in Bournez et al.'s construction of LPPs designed to compute any arbitrary algebraic number in [0,1].
\end{abstract}


\section{Introduction}\label{sec:introduction}
\subsection{GPAC/CRN Computable Number}
The computation of real numbers provides a theoretical benchmark for an abstract model's computational power. In 1936, Turing \cite{jTuri36} considered computable numbers while researching a discrete model now known as the Turing machine. Then, Shannon \cite{jShan41} introduced the general-purpose analog computer (GPAC) to the world. The GPAC is a mathematical abstraction of Bush's differential analyzer. Bush \cite{jBush31} developed this machine to obtain numerical solutions of ordinary differential equations (ODEs). After much revision \cite{jPouRic74,jLipRub87,jGraCos03,silva2004some}, we can now characterize GPACs by polynomial initial value problems (PIVPs):
\begin{equation}
    \begin{cases}
        \bx' = \bp \big(\bx(t)\big)\\
        \bx(0) = \bx_0,
    \end{cases}
    \quad t\in\R
\end{equation}
where $\bp$ is a vector of multivariate polynomials, $\bx$ is a vector of variables, and $\bx_0$ is its initial value. Bournez et al. \cite{bournez2017polynomial} have shown that GPACs are provably equivalent to Turing machines in terms of computability and complexity. Chemical reaction networks (CRNs), as a restricted form of GPACs, are widely used in molecular programming applications. The deterministic semantic of CRNs usually assumes well-mixed solutions and mass-action kinetics. CRN dynamics can be modeled by polynomial systems with extra constraints \cite{cHarTot81} (see the preliminary section for more details). Although the constraints could seemingly weaken the computational power of CRNs, they are able to simulate GPACs. Hence, the CRN model is also Turing complete \cite{cFLBP17}. The key idea is to encode every variable $x$ (possibly having negative values) in a GPAC by the difference between a pair of \emph{positive} variables $x^+$ and $x^-$ in the copycat CRN. That is, $x(t)= x^+(t) - x^-(t)$ for all $t\geq 0$.

A straightforward idea to ``compute'' a number $\alpha$ by a GPAC/CRN is to designate a variable $x$ such that $x(t)\to \alpha$ as $t\to \infty$. Further down this line, Huang et al. defined the class of real-time computable real numbers by chemical reaction networks \cite{jHKLLL19}, wherein the real-time computability notion essentially requires $x(t)$ to converge exponentially fast to $\alpha$. Huang et al. then demonstrate that GPACs and CRNs compute the same class of real numbers in real-time \cite{huang2019real}. Notably, famous transcendental numbers like $e$ and $\pi$ are among this class. Later, Huang \cite{huang2020chemical} added Euler's $\gamma$ and Dottie number (the unique real root of the equation $\cos(x)=x$) to the list. Note that in order to get a meaningful measure of time, in Huang et al.'s work, all variables in ODEs associated with GPACs or CRNs are \emph{bounded}. Hence, the notion of time aligns with the time parameter $t$ of the ODEs since no more than a linear speedup is allowed. This notion prevents the so-called Zeno phenomena \cite{jGrCaBu08} caused by inordinate speedup. Another, perhaps equivalent, metric is the \emph{arc length} of $\bx(t)$. A \emph{complexity theory} \cite{bournez2017polynomial} of analog computation can then be built upon this metric.

Also, note that the notion Huang et al. used has a \textit{descriptive} nature: the initial values are all zero, and rate constants are integrals (rationals). Hence, one can only encode a finite amount of information in the initial value and rate constant. In this sense, a system that computes a number $\alpha$ also gives a finite description of $\alpha$. This idea will prove crucial for later discussions.

\subsection{Large Population Protocols} \label{test}
Population protocols (PPs) \cite{jAAER07} can be treated as special CRNs with two inputs, two outputs, and a unit rate constant per reaction. Observe that population protocols \textit{conserve} the population since interactions do not change agent counts. This property imposes a vital constraint on the population: the sum of all agents remains constant. After normalization, we can say that this sum is \emph{one} without loss of generality.

Classical PPs are concerned with computing predicates over their input configuration in a discrete, stochastic setting. Therefore, computing real numbers does not make sense under this setting. In a series of papers \cite{aupy2011number,bournez2009convergence,cBoFrKo12}, Bournez et al. developed a setting they termed a Large-Population Protocol (LPP). LPPs are a new analytical setting rather than a new computational model. They are still population protocols, but their behaviors are analyzed as the population grows to infinity. Bournez et al. aimed for asymptotic results independent of the initial configuration. Our approach departs from theirs in that dependence on initial values is considered. Note that we should approach with caution when defining what is meant by ``initial configuration'' since LPPs need to talk about different population sizes along the way.

The LPP setting provides a playground for computing real numbers by population protocols: a number $\nu \in [0,1]$ is said to be computable by an LPP if the proportion of agents in a set of \emph{marked states} converges to $\nu$ over time as the population grows to infinity. An additional requirement is that the ODE induced by the \emph{balance equation} of the LPP must have finitely many equilibria. This condition enforces (exponential) stability and detaches the dependence on initial values.

Bournez et al. specifically asked\cite{cBoFrKo12} ``... is it possible to compute solutions of trigonometric equations? E.g., can
we design a protocol insuring that, asymptotically, a ratio $\frac{\pi}{4}$ of the nodes are in some prescribed state?'' Their answer is \emph{negative}. They proved that, under their definition, LPPs compute exactly the algebraic numbers in [0,1]. Closely related is a result by Fletcher et al. \cite{fletcher2021robust}, where a class of numbers, called Lyapunov numbers, are found to be the set of algebraic numbers. It is not a coincidence: both works involve ODEs with either finitely many or isolated equilibria. Hence, Tarski's quantifier elimination over real closed fields can be performed. Alternatively, one can find more modern elimination algorithms involving Gr\"obner bases in standard algebraic geometry texts \cite{cox2013ideals,smith2014introduction}.

Why is there a gap between LPPs (cannot compute transcendentals) and GPAC/CRN (can compute $e$ and $\pi$)? The following example helps explore the question.
\begin{motivating-example}
Let $F(t)= \frac{1}{2}e^{e^{-t} -1}$, $E(t)= \frac{1}{2}e^{-t}$, and $G(t)$ be a function such that its derivative ``cancels'' with $E$ and $F$'s derivative. We have the first-order system and the corresponding chemical reaction network:
\begin{equation*}
   \hspace{1.5 cm}
\text{ODE:}\hspace{0.5cm}
       \begin{cases}
        F'= -2FE\\
        E'= -E \\
        G'= 2FE + E
   \end{cases}
   \hspace{1 cm}
\text{CRN/PP:} \hspace{0.5 cm}
\begin{cases}
    \reaction{F + E}{2}{G+E}\\
    E \to G.
   \end{cases}
\end{equation*}
\end{motivating-example}
It is evident that with initial values $F(0)=\frac{1}{2}$, $E(0)=\frac{1}{2}$, and $G(0)=0$, the solution of $F(t)\to\frac{1}{2}e^{-1}$ as $t\to \infty$. Correspondingly, in the LPP setting, if we let $F(0)=\floor{\frac{N}{2}}$, $E(0)=\floor{\frac{N}{2}}$, and $G(0)=0$ for a population of size $N$, we can observe from the experimental results in Figure~\ref{fig:e-computation} that the proportion of $F$ gets very close to $\frac{1}{2}e^{-1}$ when $N$ becomes larger. In fact, the so-called Kurtz's Theorem \cite{jKurt72}, which says ODE solution $F(t)$ equals the proportion random variable $F^{(N)}(t)/N$ almost surely when $N$ goes to infinity, guarantees the limit is $\frac{1}{2}e^{-1}$. Although the CRN in the above example is not technically a PP, we can translate it into a probabilistic PP (Construction~\ref{construction:PLPP}, see also \cite{doty2021ppsim}, Theorem 1). We maintain the protocol in the current form for simplicity.

\begin{figure}[tbh!]
  \centering
  \includegraphics[width=\textwidth]{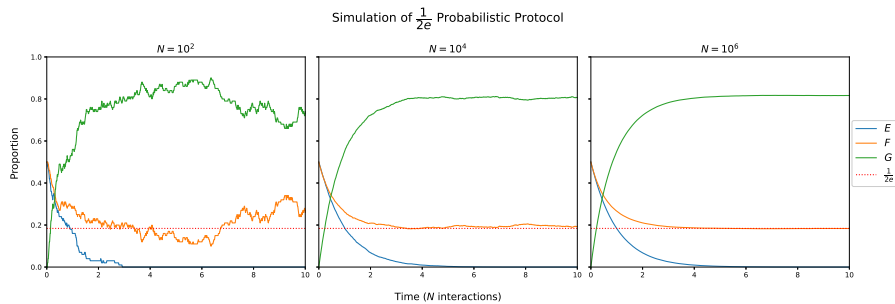}
  \caption{We use the ppsim Python package \cite{doty2021ppsim} to simulate our motivating example (transformed into a probabilistic PP). This figure demonstrates that the transcendental number $\frac{1}{2}e^{-1}$ can be traced by the proportion of a species as the population size grows.}
 \label{fig:e-computation}
\end{figure}

So the above LPP can ``compute'' $\frac{1}{2}e^{-1}$, a \textit{transcendental} number! Is it a counter-example to Bournez et al.'s algebraic result? Not really. Since the system of ODEs in the example has the whole $F$-$G$ plane (with $E=0$) as equilibria, the system has a \textit{continuum} of equilibria. Note that the key reason GPACs/CRNs in \cite{huang2019real} can compute transcendentals is that they are systems with a continuum of equilibria.
\subsection{Computing with a Continuum of Equilibria}
The Motivating Example invites us to \textit{extend} the notion of LPP-computable numbers: We can drop the finitary requirement on equilibria and therefore include systems with a continuum of equilibria. The seminal work \cite{bhat2003nontangency} by Sanjay and Bernstein provides an in-depth study of this type of system.

The move has several consequences. An immediate one is that computation now \emph{depends} on initial values: since the equilibria ``stick together'', we need the initial values to determine which equilibrium the system converges to. Another one is stability: we lose asymptotic stability. As pointed out in \cite{bhat2003nontangency}, ``... Since every neighborhood of a non-isolated equilibrium contains another equilibrium, a non-isolated equilibrium cannot be asymptotically stable. Thus asymptotic stability is not the appropriate notion of stability for systems having a continuum of equilibria.''
However, it is inaccurate to immediately conclude that such systems must be fragile and sensitive to perturbations. Indeed, Sanjay and Bernstein investigated \emph{semistability}, a more appropriate notion of stability for these systems.

We now state our new finding under the extended notion of LPP-computable numbers.
\begin{main-theorem}
    LPPs compute the same set of numbers in [0,1] as GPACs and CRNs.
\end{main-theorem}
Our proof is constructive: We give an explicit algorithm to translate a bounded GPAC or CRN into an LPP. A critical step is the construction of a cubic form (homogeneously degree 3 polynomial) system that conserves the population. Then, we transform the system from the cubic form into the quadratic form. Finally, the system can be rewritten as a probabilistic population protocol.

The key feature of our construction is that we always guarantee the system is either \emph{CRN-implementable} or \emph{PP-implementable}. Bournez et al. constructed a system in \cite{cBoFrKo12} that claimed to compute any algebraic number. While their system is a GPAC, it is not a PP in general. As a by-product, our construction fixes this flaw.

\subsubsection*{Related work:}
Besides Bournez et al.'s work \cite{cBoFrKo12,aupy2011number,bournez2009convergence}, Gupta, Nagda, and Devaraj \cite{gupta2007design} develop a framework for translating a subclass
of differential equation systems into practical protocols for distributed systems (stochastic CRNs in most cases and, in some examples, population protocols). They also restated Kurtz's theorem and considered a large-population setting. The subclasses they considered are \textit{complete} and \textit{completely partitionable}, which means the derivatives of the system sum to zero, and any term (monomial) $T$ occurs as a pair $\{+T, -T\}$ in the system, a very nice property to have. Doty and Severson \cite{doty2021ppsim} provided an algorithm to translate CRNs that consist solely of unimolecular ($X\to Y$) and bimolecular reactions ($X + Y \to Z + W$) to PPs. Those reactions correspond to an ODE system's linear terms (e.g., $x$) and quadratic terms (e.g., $xy$). Their assumption, however, restricts the application of their algorithm from more general GPACs. Firstly, the CRNs they considered preserve populations. \textit{Constant terms} in GPACs , which correspond to reactions like $\emptyset \to X $, \textit{change} the population. Another consideration is the treatment of square terms. Some quadratic systems can not be \emph{directly} written as a CRN with bimolecular reactions (i.e. if positive square terms like $x^2$ occur in $x'$). Later in this paper, we will discuss how the deep entanglement of the two (constant terms and squared terms) makes translating a general GPAC into a PP very difficult.

Our construction does not make any assumptions in the above work on GPACs, except that all variables are bounded.


The rest of this paper is organized as follows.
\Cref{sec:preliminaries} discusses preliminaries as well as some notations and conventions necessary to present our main result;
\Cref{sec:e_of_lpp_and_gpacs} includes the main theorem, a four-stage constructive proof, and a few immediate corollaries.
\Cref{sec:conclusion} provides some concluding remarks and discussion on future directions.

\section{Preliminaries}\label{sec:preliminaries}
We review some key definitions and basic facts in this section.

\subsection{GPACs, CRNs, Population Protocols, and Their ODE Characterization}
We have discussed that GPACs can be characterized by polynomial initial value problems (PIVPs) in the previous section. To align with the ODE descriptions of GPACs, we will use the ODE characterization of CRNs and PPs throughout this paper. This choice clarifies the discussion of our construction and proof in the coming sections.


We start with a few notations and conventions in this paper.

\begin{notation}[Variable Vector]
    We use the notation $\bx(t)=(x_1(t), x_2(t), \cdots, x_n(t))$ to denote a vector of variables in (usually the solution of) an ODE system, where $n\in \N$ is the vector's dimension and the parameter $t$ is regarded as ``time'' in the system. The first component $x_1$ is usually designated for tracing the number we wish to compute.
\end{notation}

A typical definition of a CRN or PP specifies a pair $(S, R)$, with $S$ being the set of species or states and $R$ being the set of reactions or interactions between states. In general, a reaction for an abstract CRN looks like $X + Y  \goesto{k} X + W + Z$, where $k$ is the \textit{reaction constant}. There is no restriction on the number of reactants (input) or products (output). Reactions in population protocols, however, are strictly limited to those with two inputs and two outputs: $  X + Y  \to W + Z.$

The so-called \textit{deterministic mass-action} semantic of a CRN defines a polynomial differential system, which governs the CRN's dynamics. One can find descriptions of this semantic in many works of literature (\eg, \cite{huang2019real}). We will use these polynomial differential systems directly throughout the paper.

\begin{convention}
  We manipulate our use of the variable vector: when the context is clear, we simply say $(\bx, \bx')$, or $\bx$ itself is a GPAC (CRN, PP), without specifying the associated polynomial differential system $\bx'(t)=(x_1'(t),\cdots, x_n'(t))$. If we don't specify the initial values that means we assume them to be zero. We would assume the correspondence between a variable $x$ and the species $X$ that it describes. In some cases, we will call a variable $x$ a species and $\bxn$ a set of species.
\end{convention}
Now, we introduce notation and terminology for polynomials.
\begin{notation}[Positive Polynomials]
    We denote $\pospoly \subset \Q[\bx]=\Q[x_1, \cdots, x_n]$ the set of multivariable polynomials with positive (rational) coefficients for each monomial (terms). For instance, $x_1x_2 +x_3 +1 \in \pospoly$, while $x_1x_2 -x_3 +1 \not\in \pospoly$.
\end{notation}
\begin{notation}[Variable Occurrence in a Polynomial or Monomial]
Let $p$ be a polynomial (or monomial) and $x$ be a variable, we use the convenient notation $x \in p$ to express $x$ occurs in $p$.
\end{notation}
 Homogeneous polynomials play a fundamental role in our construction. We call them forms for short.
\begin{term}[Forms]
   In mathematics, ``form'' is another name for homogeneous polynomials. A \textit{quadratic form} is a homogeneous polynomial of degree two, for example, $4x^2 + 2xy - 3y^2$; a \textit{cubic form}  is a homogeneous polynomial of degree three, for example, $x^{3}+3x^{2}y+3xy^{2}+y^{3}$.
\end{term}

\begin{term}[Conservative Systems]
    A conservative system has a (non-trivial) quantity that is constant along each solution. In this paper, we call a system that maintains constant total mass, \ie, a system $\bxn$ that satisfies
    \[
        \sum_{i=1}^{n}x_i'=0,
    \]
    a conservative system.
\end{term}

We say a function $f(t)$ is implementable (generable) by a GPAC (CRN, LPP) if there exist a GPAC (CRN, LPP) and a variable (species, state) $x$ in it, such that $f(t) = x(t)$ for all $t$. We choose the word \textit{implementable} over \textit{generable} (as in \cite{jGraCos03}) to reflect the idea that we \textit{program} a GPAC (CRN, LPP) to \emph{implement} a function.
\begin{definition}[GPAC-implementable Functions \cite{jGraCos03}. See also Definition 3 in \cite{cFLBP17}.]
A function $f:\R_{\geq 0}\to R$ is GPAC-implementable if it is the first component \ie $x_1$, of the $\bx(t)=(x_1, x_2, \cdots, x_n)$ solution for the following differential equation:
\begin{equation}
    \begin{cases}
        \bx' = \bp \big(\bx(t)\big)\\
        \bx(0) = \bx_0,
    \end{cases}
    \quad t\in\R
\end{equation}
where $\bp$ is a vector of multivariate polynomials and $\bx$ is a variable vector with $\bx_0$ as its initial value.
\end{definition}
Note that if a variable $x$ has initial value $a$, we can introduce a new variable $y= x-a$ to make $y(0)=0$. Therefore, without loss of generality, we often assume $\bx(0) =\mathbf{0}$ in this paper.

\begin{theorem}[CRN-implementable Functions \cite{cHarTot81}, Theorem 3.1 and 3.2]
A function is CRN-implementable if and only if it is GPAC-implementable with further restriction such that the derivative of each component $x_i$ in its variable vector $\bxn$ is of the form
\begin{equation}\label{eq:CRN-char}
     x_i' = p - q x_i, \quad \text{where }p,\ q \in\pospoly.
\end{equation}
\end{theorem}
The restriction is rooted in reaction modeling: a negative term in $x'$ means to \textit{destroy} some molecular $x$ at some rate. However, a reaction cannot destroy a non-reactant, so $x$ must appear as a reactant in the reaction. Such restriction extends to Population Protocols together with some extra constraints.
\begin{corollary}[PP-implementable Functions]
\label{thm:PP-cond}
A function is PP-implementable, if and only if
\begin{enumerate}
    \item[i.] it is CRN-implementable, \ie, there is a CRN $(\bx,\bx')$ computing it.
    \item[ii.] for all $x_i$ in $\bxn$, $x_i'$ does not possess a positive $x_i^2$ term.
    \item[iii.] $\bx'$ is a quadratic form system, and
    \item[iv.] $\bx'$ is conservative.
\end{enumerate}
\begin{proof}
 For the ``only if'' direction, all conditions other than condition (ii) are clear. To see that $x_i'$ can not have an $x_i^2$, we consider two $X_i$'s on the left-hand side of a reaction. The reaction can not increase the amount of $X_i$, since that would require at least $k+1$ $X_i$'s as products on the right-hand side. As a result, a positive $x_i^2$ term can not occur in $x_i'$.
 See Theorem \ref{thm:PLPP} for the ``if'' direction of the proof.

\end{proof}
\end{corollary}
In this paper, we will also call a CRN containing only unimolecular reactions of the form $X \to Y$ a \textit{unimolecular} population protocol (UPP). Similarly, a CRN containing only termolecular reactions of the form $X + Y + Z \to U + V + W$ is called a \textit{termolecular} population protocol (TPP). In general, a $k$-PP is a CRN that contains only reactions of the form $X_1 + \cdots + X_k \to  Y_1 + \cdots + Y_k$.
Generally, we have the following characterization.
\begin{corollary}[$k$-PP-implementable Functions]
  \label{cor:kPP}
A function is $k$-PP-implementable, if and only if
\begin{enumerate}
    \item[i.] it is CRN-implementable, \ie, there is a CRN $(\bx,\bx')$ computing it.
    \item[ii.] for all $x_i$ in $\bx$, $x_i'$ does not possess a positive $x_i^k$ term.
    \item[iii.] $\bx'$ is homogeneously degree $k$, and
    \item[iv.]  $\bx'$ is conservative.
\end{enumerate}

\end{corollary}

From above, we see that turning a GPAC into a PP means having to dance with more and more chains. We must carefully appease all restrictions, which makes the translation process difficult.

We first revisit some useful concepts regarding PPs.
\subsection{Probabilistic Large-Population Protocols}

\begin{definition}[Balance Equations \cite{cBoFrKo12}]
    Let $(Q, R)$ be a large-population protocol, where $Q$ is the state (species, variables) set and $R$ is a set of reaction rules, each with two reactants and two products. The \emph{balance} equation of a \textit{deterministic} LPP is a function $b:\mathbb{R}^{|Q|}\to\mathbb{R}^{|Q|}$ such that
\begin{equation}\label{eq:balance_eq_det}
    b(x)=\sum_{(q_1, q_2)\in Q^2}\bigg( x_{q_1} x_{q_2}\big(-e_{q_1}-e_{q_2}
    + \sum_{(q_3, q_4)\in Q^2} \delta_{q_1, q_2,q_3, q_4}(e_{q_3}+ e_{q_4})\big)\bigg)
\end{equation}
where $\delta_{q_1, q_2,q_3, q_4}=1$ if $(q_1, q_2) \to (q_3,q_4)$, and 0 otherwise; $(e_q)_{q\in Q}$ is the canonical base of $\mathbb{R}^{|Q|}$ and $x_{q_i}$ is the proportion of the population in state $q_i$.
\end{definition}
The balance equation is a description of the system's dynamics. Intuitively, it says whenever $q_1$ and $q_2$ bump into each other, if there is a reaction $(q_1, q_2) \to (q_3,q_4)$ in the PP, then the pair $(q_1,q_2)$ turns into $(q_3, q_4)$; otherwise, we interpret $(q_1, q_2)$ as a \emph{null reaction} and the agents remain unchanged. We adopt an \emph{asymmetric} interpretation of reactions in PPs throughout the paper. That is, the ordered pairs $(q_i , q_j ) $ and $(q_j , q_i )$, when used as reactants, do not necessarily result in the same products. The assumption is not essential: the choice is to align with the existing literature (\eg, \cite{cBoFrKo12}) for convenience.

The above balance equation is for deterministic PPs, where $\delta_{q_1, q_2,q_3, q_4}$ is either 0 or 1. It is more convenient to use probabilistic transition rules in many scenarios. We introduce probabilistic LPPs (PLPPs).

\begin{definition}[Probabilistic LPP (PLPP) \cite{cBoFrKo12}]
A PLPP is an LPP with transition rules in the form
\[
    q_i \ q_j \to \alpha_{i,j,k,l}\ q_k \ q_l
\]
and for every $(q_i,q_j)\in Q^2$, we have
\begin{itemize}
    \item for every $(q_k, q_l)\in Q^2$, $\alpha_{i,j,k,l}\in\Q$ and $\alpha_{i,j,k,l}>0$, and
    \item $\displaystyle\sum_{(q_k,q_l)}\alpha_{i,j,k,l} =1.$
\end{itemize}
\end{definition}
 The balance equation remains mostly the same as Equation~(\ref{eq:balance_eq_det}). To simplify the notations, we often take $Q =[n]=\{1,2,\cdots,n\}$ and will still use $q_i$ or use terms like ``state $i$'' when referring to a state.
\begin{equation}\label{eq:balance_eq_prob_simple}
    b(x)=\sum_{(i, j)\in [n]^2}\bigg( x_{i} x_{j}\big(-e_{i}-e_{j}
    + \sum_{(k, l)\in [n]^2} \alpha_{i, j, k, l}(e_{k}+ e_{l})\big)\bigg)
\end{equation}
The ODE associated with a PLPP can be written as
\begin{equation}\label{eq:balance-equation}
    \frac{\diff \bx}{\diff t}= b(\bx),
\end{equation}
where $\bx=(x_1(t),\cdots, x_n(t)) \in \R^{n}$ is a variable vector and its $i$-th component, $x_{i}$, tracks the proportion of the total population in state $i$ over time $t$. We will often call $\bxn$ or $\bx$ an LPP (PLPP).

A simple observation results directly from the definition of LPP or PLPP.
\begin{observation}[``The one trick'']\label{obs:one-trick}
Let $\bxn$ be an LPP or PLPP, then $\displaystyle \sum_{i\in [n]} x_i = 1.$
\end{observation}
Although obvious, the above is the single \emph{most important} observation in this paper. Many constructions and proofs rely on it. This observation is frequently used to re-write the constant term in an ODE. We will rewrite a constant $c$ in either of the following ways, depending on our needs:
\begin{itemize}
    \item $c= c\cdot 1= c(x_1+\cdots+x_n)$, or
    \item $c=c\cdot 1 \cdot 1 =c (x_1+\cdots+x_n)(x_1+\cdots+x_n)$.
\end{itemize}
\begin{observation}\label{obs:component}
  The $r$-th component of Equation~(\ref{eq:balance-equation}) has the form
\begin{equation}\label{eq:component_balance}
  \frac{\diff x_{r}}{\diff t}= f(\bx)-2x_{r},
\end{equation}
where $f(\bx)$ is a quadratic form in $\pospoly$.
\end{observation}
An intuitive way to think about equation \ref{eq:component_balance}:
\begin{itemize}
    \item The term $-2x_r$ represents all reactions that \emph{consume} $x_r$. A complete set of rules for a PP must include all combinations of pairs $(x_r, x_k)$, or $(x_k, x_r)$, where $k$ iterates through all $[n]$. These combinations sum to $-2 x_r$.
    \item The $f(x)$ represents all reactions that \emph{produce} $x_r$. This observation is critical in the constructive proof of our main theorem.
\end{itemize}

\begin{proof}
   Take the $x_r$ component in Equation~(\ref{eq:balance-equation}). We have
   \begin{align*}
    \frac{\diff x_{r}}{\diff t}
    &=\sum_{ j\in [n]}\bigg( -x_{r} x_{j}\bigg)+ \sum_{j\in[n]}\bigg( -x_{r} x_{j}\bigg)\\
    &\quad + \textstyle\proj_r\Bigg(\sum_{i,j\in [n]}\bigg( x_{i} x_{j}\sum_{k, l\in [n]} \alpha_{i, j,k, l}(e_{k}+ e_{l})\bigg)\Bigg),  \\
    &= -2x_r  + f(\bx), \quad\text{by Observation \ref{obs:one-trick}.}
   \end{align*}
   where  $f(\bx)= \textstyle\proj_r\Bigg(\sum_{i,j\in [n]}\bigg( x_{i} x_{j}\sum_{k, l\in [n]} \alpha_{i, j,k, l}(e_{k}+ e_{l})\bigg)\Bigg)$ and $\proj_{r}(\cdot)$ is the $r$-th component a vector. This completes the proof.
\end{proof}

We now give a formal definition of LPP-computable numbers.

\subsection{Large-Population Protocol Computable Numbers: Revisited}\label{ssec:LPP-revisit}

We extend the definition from \cite{cBoFrKo12}. Note that the definition also applies to large-population \emph{unimolecular} protocols (LPUPs).
\begin{definition}
    A real number $\nu$ is said to be computable by an LPP (PLPP, LPUP) if there exists an LPP (PLPP, LPUP) such that $\bx(t)=(x_1(t),x_2(t),\cdots,x_n(t))\in [0,1]^n$ and
    \[
        \lim_{t\to \infty}\sum_{i\in M}x_i(t)=\nu,
    \]where $M\subseteq\{1,\cdots, n\}$ represents the subset of states marked in $\bx$. Moreover, all the states $x_i$ must be initialized to some positive rational $r_i\in \mathbb{Q}\cap[0,1]$, in the sense that $\lim_{N\to \infty} x_i^{(N)}(0)=r_i$, when $x_i^{(N)}(0)$ is the initial fraction of state $i$ at the stage when the population is $N$.
\end{definition}

A major change is our removal of the finitary requirement on equilibria. A note originally in \cite{cBoFrKo12} discusses the rationale for the requirement of finitely many equilibria. This assumption is needed to
\textit{``avoid pathological cases, in particular the  case  of  idle  systems $q\  q' \rightarrow q\  q'$ for  all $q$ and $q'$.   Indeed, in  idle  systems,  all initial states are equilibria, and such a system would compute any real of [0,1], depending on the initial configuration.''}

 This pathological behavior could occur, if one is allowed to initialize the system with any real numbers. However, we can prevent this case by restricting the initial configuration to rational numbers. Also note that by forcing the initial counts of a species $X$ to be a fraction of the total population $N$, we exclude those systems (\eg,\cite{dudek2018universal,alistarh2017robust}) that crucially depend on small counts in a very large population. These systems are not modeled correctly by ODEs. For instance, one can not initialize $X$ to have 100 molecules. Any state initialized to a constant amount will be a fraction that approaches zero, as $N$ grows to infinity.  In our definition, such $X$ will be initialized to have \textit{zero} molecules for a fixed $N$.

\begin{remark*}
    It is worth noting the above setup satisfies the assumptions of Kurtz's Theorem \cite{jKurt72}. Therefore, the long-term behavior of an LPP is governed by an ODE, which is induced by its \textit{balance equation}. As a result, in this paper, we often blur the  boundary between the continuous semantics (by ODE) and stochastic semantics when discussing limiting behavior. See Section \ref{sec:conclusion} for initial discussion on further investigation on the two semantics.
\end{remark*}

For the rest of the paper, whenever we mention LPP computable number, we mean the extended notion defined above, unless specified otherwise.
\subsection{Basic Results}
\subsubsection{Large-Population UPPs Compute Only Rationals}
As a toy case study, we investigate the class of numbers computable by large-population unimolecular protocols (LPUP). We find that unimolecular protocols have limited power.
\begin{lemma}
A number is LPUP-computable if and only it is rational.
\end{lemma}

\begin{proof}

  Suppose $\nu=\frac{p}{q}\in[0,1]$ is a rational number. Lemma 3 of Bournez et al. \cite{cBoFrKo12} constructs an LPP with $q$ states such that each state converges $\frac{1}{q}$ for an exponentially stable equilibrium. We modify their construction to obtain a unimolecular LPP. Let $X=\{X_1,\dots,X_q\}$ be the state set and the transitions be of the form $X_i \rightarrow X_{(i \mod q) + 1}$  (see Figure~\ref{fig:visual-unimolecular} for a visualization). Then each state will converge to $\frac{1}{q}$. To compute $\nu$, mark states $X_1,\dots, X_p$. Thus, there exists a deterministic unimolecular LPP that computes $\nu$. (Alternatively, one could initialize an idle system with the desired rationals, since now the extended notion of LPP-computable number allows a continuum of equilibria and rational intial values.)

  Suppose a deterministic unimolecular protocol $\mathcal{P}$ computes some number $\nu\in[0,1]$. Then there is a graph $G$ corresponding to the reactions of $\mathcal{P}$ such that nodes are states and directed edges are transitions. The deterministic nature ensures that self-loops exist only if a node is isolated from all other nodes. Since each node has the outdegree $1$, $G$ is a functional graph. Functional graphs are a set composed of rooted trees attached to a cycle (\cite{flajolet2009analytic}, page 129; or see Figure~\ref{fig:functional-graph} for an example). Traversing the graph from any starting point will eventually lead to being stuck in a cycle. The proportion of mass that flows into a cycle will become uniformly distributed among the cycle's finite number of nodes, computing a rational number. Thus, deterministic unimolecular protocols compute rational numbers.
\end{proof}
Note that the above lemma can be extended to probabilistic unimolecular protocols, but we omit this work for the sake of simplicity.

\begin{figure}[!htb]
    \centering
    \begin{minipage}{.5\textwidth}
      \centering
      \includegraphics[scale=0.7]{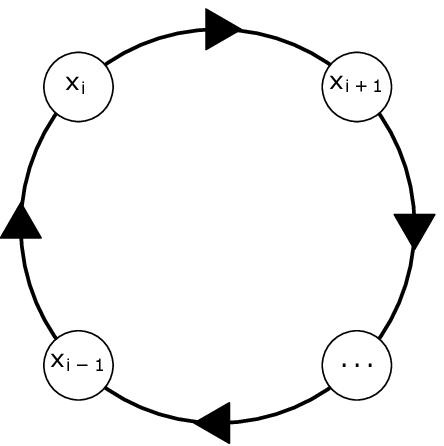}
      \caption{Visualization of transitions in a unimolecular protocol.}
      \label{fig:visual-unimolecular}
    \end{minipage}%
    \begin{minipage}{0.5\textwidth}
      \centering
      \includegraphics[scale=0.25]{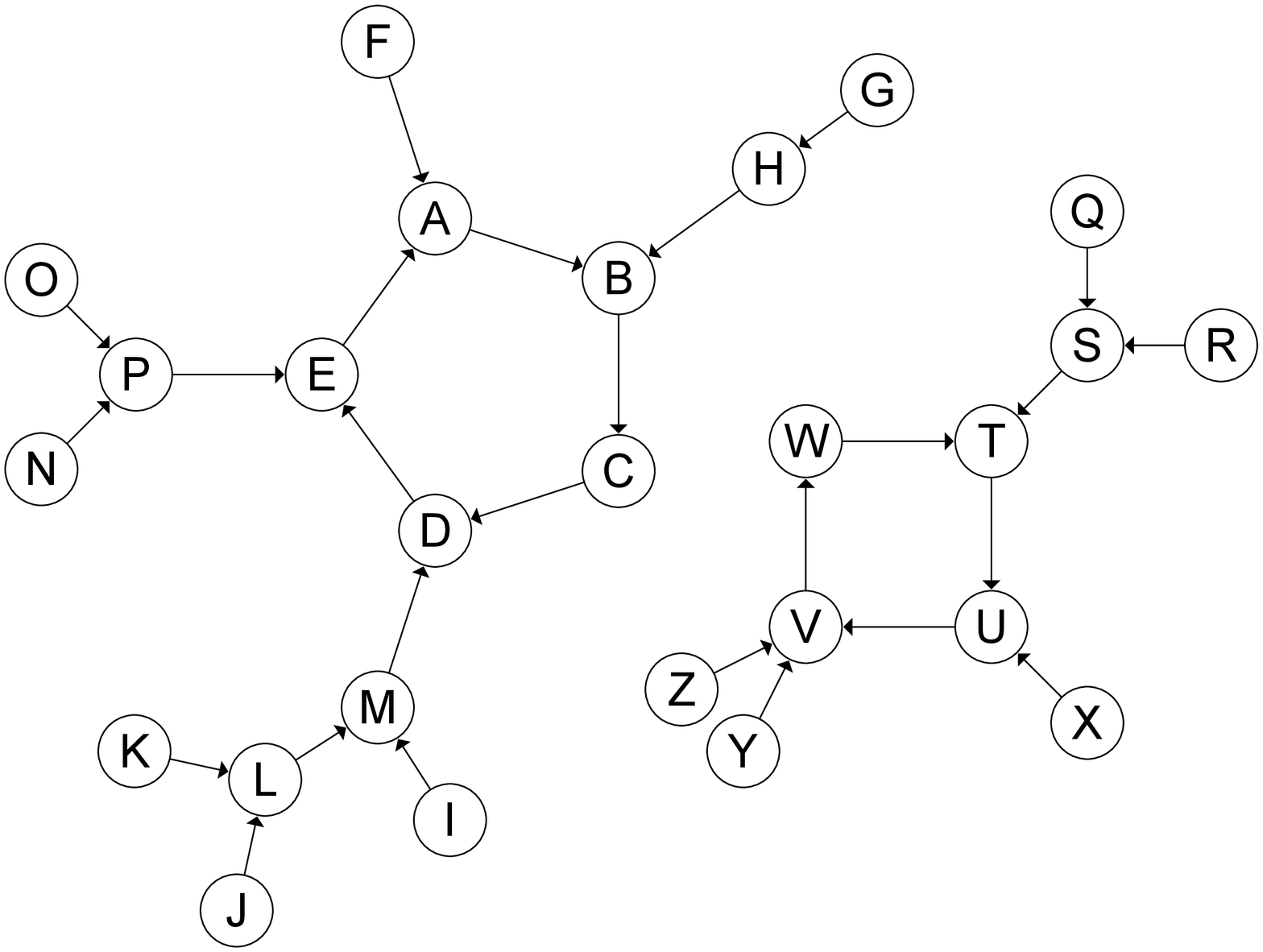}
      \caption{Example of a Functional Graph.}
      \label{fig:functional-graph}
    \end{minipage}
\end{figure}

\subsubsection{Product Protocols and Multiplication}
We will show that LPP-computable numbers are closed under multiplication. The proof is done through product protocols. Since our Main Theorem shows that LPP-computable numbers are essentially GPAC-computable, many other arithmetic operations can be done by taking a detour via GPACs. However, the resulting protocols are not so intuitive.  In addition, the product protocol is an important technique for our later construction in Section \ref{sec:e_of_lpp_and_gpacs}.
\begin{figure}
  \centering
  \includegraphics[width=\textwidth]{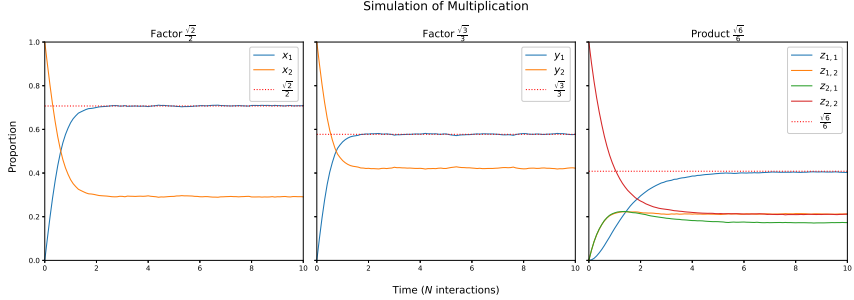}
  \caption{With ppsim, we can simulate a product protocol with input LPPs $\bx$ and $\by$ that compute $\frac{1}{\sqrt{2}}$ and $\frac{1}{\sqrt{3}}$. Note that the factor LPPs originate from \cite{bournez2009convergence}, and the product LPP is a result of applying Lemma~\ref{lemma:multiplication}. These simulations affirm the closure of LPPs under multiplication.}
 \label{fig:productProtocol}
\end{figure}
\begin{lemma}[Product of LPP-computable Numbers is LPP-computable]
\label{lemma:multiplication}
  Let ${P}_x$ and ${P}_y$ be LPPs that compute $\alpha$ and $\beta$ respectively.
  Denote their respective state vectors by $\bxn$ and $\by=(y_1,y_2,\cdots,y_m)$, initial configurations by $\bx(0)$ and $\by(0)$, and differential systems by $\bx'$ and $\by'$.
  Assume any differential equation in $\bx'$ or $\by'$ is a quadratic form.
  Then there exists an LPP, ${P}_{z}$, that computes the product $\alpha \beta$.
\end{lemma}

\begin{proof}

Let $p_{x_i},q_{x_i},p_{y_j},q_{y_j}\in\pospoly $. Then we can write $x_i'=p_{x_i} - x_i q_{x_i}$ for all $x_i'\in \bx'$. Since $x_i'$ is a quadratic form, $p_{x_i}$ is a quadratic form and $q_{x_i}$ is a linear form. This notation generalizes to $y_j'\in \by'$ as well.
Introduce new state variables $z_{i,j}:= x_i y_j$ for all $(x_i,y_j)\in \bx \times \by$ where $\times$ is the Cartesian product (view $\bx$ and $\by$ as sets). Then we have the state vector $\bz=(z_{1,1},z_{1,2},z_{2,1},\dots,z_{n,m})$.
Suppose $M_\bx \subset \bx$ and $M_\by \subset \by$ are the marked states. For each $z_{i,j}=x_i y_j \in \bz$, mark $z_{i,j}$ if $x_i\in M_\bx$ and $y_j\in M_\by$. For the initial values, let $\bz(0)=\bx(0)\cdot\by(0)$.

To obtain the differential system $\bz'$, perform the following steps. First, differentiate $z_{i,j}=(x_i y_j)$. By the chain rule, we have
\[z_{i,j}' = (x_i y_j)' = x_i'y_j + x_iy_j'.\]
Next, substitute the equations for $x_i'$ and $y_j'$. Then, we have the cubic form,
\[z_{i,j}' =y_j (p_{x_i}-x_i q_{x_i}) + x_i (p_{y_j}-y_j q_{y_j}).\]
Apply the one-trick (Proposition~\ref{obs:one-trick}) to obtain a homogeneously degree four polynomial such that every term is exactly degree two in variable $x$ and exactly degree two in variable $y$. We have,
\[z_{i,j}' = y_j\left( \sum\limits_{k=1}^m y_k \right)(p_{x_i}-x_i q_{x_i})+ x_i \left(\sum\limits_{k=1}^n x_k \right) (p_{y_j}-y_j q_{y_j}).\]
After expanding the polynomial, group the positive and negative terms. Note that we can factor $x_i y_j$ from the collection of negative terms, which yields
\[	z_{i,j}' = \left(\sum\limits_{k=1}^m y_k y_j p_{x_i} +  \sum\limits_{k=1}^n x_k x_i p_{y_j} \right)  - (x_i y_j) \left( q_{x_i}\sum\limits_{k=1}^m y_k+ q_{y_j} \sum\limits_{k=1}^n x_k \right) .
\]
By substituting $z_{i,j}$ for the factored $x_i y_j$, we have
\[	z_{i,j}' = \left(\sum\limits_{k=1}^m y_k y_j p_{x_i} +  \sum\limits_{k=1}^n x_k x_i p_{y_j} \right)  - z_{i,j} \left( q_{x_i}\sum\limits_{k=1}^m y_k+ q_{y_j} \sum\limits_{k=1}^n x_k \right),
\]
which satisfies condition $(i)$ from Corollary~\ref{thm:PP-cond}.

Let $p_{z_{i,j}}= \left(\sum_{k=1}^m y_k y_j p_{x_i} + \sum_{k=1}^n x_k x_i p_{y_j} \right)$.
Without loss of generality, consider $p_{z_{i,j}}^*=\sum_{k=1}^m y_k y_j p_{x_i} $.
Note that every monomial in $p_{z_{i,j}}^*$ is of the form $x_u x_v y_j y_k$ where $x_u x_v$ is a monomial in $p_{x_i}$.
Then we can substitute either $z_{u,j}z_{v,k}$ or $z_{u,k} z_{v,j}$ for each  $x_u x_v y_j y_k$ in $p_{z_{i,j}}^*$.
Since $x_i^2$ is not a valid monomial of $p_{i_x}$, we will not have $z_{i,j}^2$ in ${p_{z_{i,j}}^*}$.
Thus, $p_{z_{i,j}}$ is a quadratic form in variable $z$ that satisfies condition $(ii)$ of Corollary~\ref{thm:PP-cond}.

Let $q_{z_{i,j}}=\left( q_{x_i}\sum_{k=1}^m y_k+ q_{y_j} \sum_{k=1}^n x_k \right)$.
Without loss of generality, consider $q_{z_{i,j}}^*=q_{x_i}\sum_{k=1}^m y_k$.
Note that every monomial in $q_{z_{i,j}}^*$ is of the form $x_u y_k$, where $x_u$ is a monomial in $q_{x_i}$.
Then we can substitute $z_{u,k}$ for each  $x_u y_k$ in $q_{z_{i,j}}^*$.
Thus, we can rewrite $z_{i,j}'$ as
\[ z_{i,j}' = p_{z_{i,j}} - z_{i,j} q_{z_{i,j}},
\]
where $z_{i,j}'$ is a quadratic form and meets all conditions of Corollary~\ref{thm:PP-cond}.

Hence, we have defined the state set, initial configuration, and differential system of a PP-implementable protocol that computes $\alpha\beta$.
\end{proof}

\section{Translating Bounded GPACs into PPs}\label{sec:e_of_lpp_and_gpacs}
\subsection{Construction Overview}
We will prove our main theorem in this section. It suffices to show that LPPs can simulate GPACs. Since one can transform a bounded GPAC into a bounded CRN with zeroed initial values \cite{huang2019real}, it is sufficient to consider only this type of CRN. We develop a four-stage algorithm that translates CRNs into LPPs.
\begin{itemize}
    \item Input: A CRN that computes a target number $\alpha$. We now enter a cascade of transformations:
    \item Stage 1: into a CRN-implementable quadratic system,
    \item Stage 2: into a TPP-implementable cubic form system,
    \item Stage 3: into a PP-implementable quadratic form system, and
    \item Stage 4: into a PLPP.
\end{itemize}

\subsection{Basic Operations}
 First, we discuss some basic operations as building blocks for our constructive proof. The notations used are as described in the preliminary section. The naming of operations ``$\varepsilon$-trick'' and ``$\lambda$-trick'' (the use of the Greek letters $\varepsilon$ and $\lambda$) originates in the proof of Lemma 5 in \cite{cBoFrKo12}.
\begin{operation}[Constant Dilation, ``$\varepsilon$-trick'']\label{op:ep-trick}
Let $\bx(t)$ be the solution to the ODE $\bx'(t)=p(\bx(t))$ and $\varepsilon\in \Q^+$ a positive constant, then $\bx(\varepsilon t)$ is the solution to the ODE $$\bx'(t)=\varepsilon p(\bx(t)).$$
\end{operation}
The operation is essentially a change of variable $t\to \varepsilon t$ and a direct result of the chain rule. The behavior of the ODE system will not change; it only alters the flow of time. Indeed, $t$ in the new system corresponds to $\varepsilon t$ of the original system. In our construction, we will use this operation to shrink the coefficient of an ODE system. The operation can be seen in \cite{doty2021ppsim,cBoFrKo12}.

We can do a more general change of variable (composition.)
\begin{operation}[Time Dilation (\cite{parker1996implementing}, Theorem 3) ]
Let $\bx(t)$ be the solution to the GPAC $\bx(t)'= p(\bx(t))$, $g(t)$ be a GPAC-implementable function, and $G(t)=\int_0^t g(s)ds$, then $\bx(G(t))$ is the solution to the GPAC
$$\bx'(t)=p(\bx(t))\cdot g.$$
In our construction,  we must choose $g(t)$ such that its integral satisfies
\[
     G(t)\vert_{t=\infty}=\int_0^\infty g(s)ds = \infty.
\]
and we call such a $G(t)$ as ``nonterminating time''.
\end{operation}
A time dilating operation, as above, controls the flow of time with the function $g(t)$. More accurately, $t$ in the new system corresponds to $G(t)=\int_0^t g(s)ds$ of the original system. If we want the new system to mimic the old system's limiting behavior (\ie, to achieve the value $\lim_{t\to \infty}x_1(t)$), we must guarantee $G(\infty)=\infty$. Thus, time in the new system is nonterminating and can go to infinity in the eyes of the old system. Otherwise, the new system would terminate somewhere in the middle, relative to the old system. To obtain nonterminating time, it is enough to consider functions that converge slowly, \eg, $g(t)=\frac{1}{1+t}$. Alternatively, one could ensure $g$ is bounded from below, \ie, $g(t)>c>0$, for some constant $c$.

Time dilation for CRNs can be found in \cite{klinge2016modular} (page 15, Theorem 3.7).

Next, we introduce an operation that shrinks variables to [0,1].
\begin{operation}[Variable Shrinking, or ``$\lambda$-trick'']
  \label{op:lambda}
Let $\bxn$ be a bounded CRN that computes a number $\alpha$ in $[0,1]$. That is, $\lim_{t\to \infty}x_1(t) =\alpha$. To transform the CRN into a LPP, we must ensure all other components, some of which may not not converge, fall in [0,1] for all $t$. We can pick a $0<\lambda<1$ sufficiently small and a constant $c>0$ such that
\[
     x_1 + \lambda (x_2 + \cdots + x_n) < 1-c.
\]
This is possible since all variables are bounded. Then we
\begin{itemize}
    \item make a change of variables: $x_2 \leftarrow \lambda x_2, \cdots, x_n \leftarrow \lambda x_n$, and
    \item introduce a new variable $x_0(t)$ such that $x_0(t) + x_1(t)+\cdots + x_n(t) =1$ by
    \begin{enumerate}
        \item  initialing $x_0(0)=1$ (recall $x_i(0)=0$ for $i=1,\cdots,n$), and
        \item  adding the ODE for $x_0$: $x_0'=-\sum_{i=1}^{n} x_i'$.
    \end{enumerate}
\end{itemize}
Note that $x_1$ remains unchanged in the operation, so it still computes $\alpha$.
\end{operation}

The biggest challenge in translating a CRN into an LPP is: the CRN does not necessarily preserve population, while the LPP requires it. In fact, our input CRN initializes at zero. However, $x_1$ needs to compute some number $\alpha$ as all other components remain positive throughout time; evidently, not preserving the population. A straightforward idea to ``balance'' the system, is to introduce a variable $x_0$, such that its derivative cancels all changes to the population size by other variables. That is, $x'_0= -\sum_{i=1}^n x_i'$.

However, the resulting $x_0$ is not CRN-implementable in general. As a result, it is not transformable to an LPP. The rationale is clear in the derivative of $x_i$. In general, $x_i'$ has the form
\[
    x_i' = p_i - q_i\cdot x_i, \quad \text{where $p,q\in \pospoly$.}
\]
A typical $p_i$ would contain at least one monomial $m$. Evidently, $x_0 \not \in m$ since the variable $x_0$ is newly introduced. Then $m$ will contribute to a negative term in $x_0'$ and $x_0 \not \in m$, making $x_0$ not CRN-implementable and, in turn, not PP-implementable. That is why Bournez et al.'s construction in \cite{cBoFrKo12} (Lemma 2) fails to produce a population protocol in general. To address the this issue, we purposefully \emph{multiply $x_0$} with every $x_i'$ and introduce the following operation.

\begin{operation}[Balancing Dilation, or ``Garbage Collection'' (``g-trick'')]
Let $\bxn$ be a CRN-computable system, (\ie, every component $x_i$ is CRN-computable). Then the following operation (multiplying an \emph{unknown} variable $x_0$ determined by an ODE) produces a conservative CRN-implementable system.
\[
\begin{cases}
  x_i' = (p_i - q_i \cdot x_i)\cdot x_0,  \quad \text{for $i\in\{1, \cdots, n\}$}\\
  x_0' = -\sum_{i=1}^{n} x_i'.
\end{cases}
\]
\end{operation}
What makes the operation different from a normal time dilation operation is, the function $x_0$ is not known in advance: It needs to adjust/customize according to the system it is dilating. In our construction, we also need to combine the operation with the $\lambda$-trick to achieve nonterminating time. This operation also turns a non-conservative system into a conservative \textit{CRN-implmentable}  one.

We will present our construction in the following subsections.

\subsection{Stage 1: Conversion to CRN-implementable Quadratic Systems}
The following is a CRN version of \cite{carothers2005some} (Theorem 1).
\begin{theorem}\label{thm:deg_leq_2}
    Any solution of a CRN is a solution of a CRN-implementable system of degree at most two.
\end{theorem}
\begin{proof}
Let $\bxn$ be a CRN, then each $x_i$ has the form $x_i'= p_i - q_i x$, where $p_i, q_i\in \pospoly$. We in introduce variables (we call them $v$-variables) for each monomial as follows:
    \[
        v_{i_1,\cdots, i_n}= x_1^{i_1}x_2^{i_2}\cdot x_n^{i_n}.
    \]
Specially, $v_{1,0,\cdots,0}=x_1$, $v_{0,1,0,\cdots, 0} = x_2$, and so on. Note that we can rewrite $p_i$ and $q_i$ in each $x_i'$ with the $v$-variables and denote them as $P_i$ and $Q_i$. We can see that $P_i$ and $Q_i$ are of degree one in terms of the $v$-variables.

 Consider the system that contains the collection of all $v_{i_1,\cdots, i_n}$'s. Notice that
\begin{align*}
    v_{i_1,\cdots,i_k,\cdots, i_n}'&= \sum_{k=1}^{n}x_1^{i_1}\cdots x_k^{i_k -1}\cdots x_n^{i_n} \cdot x_k'\\
    &=\sum_{k=1}^{n}v_{i_1,\cdots, i_k-1, \cdots, i_n}(P_k -Q_k x_k)\\
    &= \sum_{k=1}^{n}P_k v_{i_1,\cdots, i_k-1, \cdots, i_n} -  \sum_{k=1}^{n} Q_k (x_k\cdot v_{i_1,\cdots, i_k-1, \cdots, i_n})\\
    &= \sum_{k=1}^{n}P_k v_{i_1,\cdots, i_k-1, \cdots, i_n} -  \Big(\sum_{k=1}^{n} Q_k \Big) \cdot v_{i_1,\cdots, i_k, \cdots, i_n}.
\end{align*}
Note that the $v$-system is a CRN-implementable system with degree less than two.
\end{proof}
With this result, we can assume a CRN's associated polynomial system is quadratic with out loss of generality.
\subsection{Stage 2: Conversion to TPP-implementable Cubic Form Systems}
\begin{theorem}\label{thm:tpp_implementable}
    Any solution of a quadratic CRN is a solution of a TPP-implementable cubic form system.
\end{theorem}
\begin{proof} (sketch)
Given a quadratic CRN $\bxn$, we first apply the $\lambda$-trick to ``make room'' for a new variable $x_0$. That is, we choose a $0<\lambda<1$ and a constant $0<c<1$, such that $ x_1 + \lambda (x_2 + \cdots + x_n) < 1-c.$ Then apply the one-trick and balancing dilation to get a cubic form:
\begin{enumerate}
    \item Introduce a new variable $x_0$
    \item Rewrite each $p_i$ (quadratic), $q_i$ (linear) in $x_i'= p_i-q_i\cdot x_i$ to forms:
           \begin{itemize}
            \item for a linear monomial in $p_i$ of the form $a\cdot x_j$, apply the ``one''-trick and rewrite it as $a x_j \sum_{k=0}^{n}x_k$.
            \item for a constant term in $p_i$ of the form $a$, rewrite it as $a\cdot(\sum_{k=0}^{n}x_k)\cdot(\sum_{k=0}^{n}x_k)$.
            \item for a constant term $a$ in $q_i$, rewrite it as $a\cdot\sum_{k=0}^{n}x_k$.
           \end{itemize}
         We call the resulting forms $P_i$ and $Q_i$ respectively.
    \item Apply balancing dilation with $x_0$. The new system looks like
    \[
        \begin{cases}
         x_i'= (P_i - Q_i x_i)\cdot x_0, \quad\text{for $i=1,\cdots, n$}\\
         x_0'= -\sum_{i=1}^{n}x_i'.
        \end{cases}
        \quad \text{and }
                \begin{cases}
         x_i(0)= 0, \quad\text{for $i=1,\cdots, n$}\\
         x_0(0)= 1.
        \end{cases}
    \]
    It is a TPP-implementable cubic form system. See the appendix for a full proof.
\end{enumerate}
\end{proof}
\begin{proof}(full)

Let $\bx =( x_1,\dots,x_n)$ be a quadratic CRN-implementable system. Then we have $x_i' = p_i - x_i q_i$ such that $p_i, q_i\in\pospoly$ for $i=1,\dots,n$. Since the CRN is quadratic, it follows that $\deg(p_i)\leq 2$ and $\deg(q_i)\leq 1$ for  $i=1,\dots,n$.

First, apply the $\lambda$-trick to shrink variables $x_2,\dots,x_n$. Choose $0\le \lambda \le 1$ and a constant $0<c<1$ , such that $x_1 + \lambda(x_2 + \cdots + x_n) < 1-c$. Then perform the change of variables: $x_2 \leftarrow \lambda x_2, \cdots, x_n \leftarrow \lambda x_n$.

Next, we introduce a new variable $x_0$ such that $\sum_{i=0}^n x_i = 1$. Use this summation and the ``one-trick'' to ensure that every term in $x_i'$ is quadratic for $i=1,\dots,n$. Once each $x_1',\dots,x_n'$ is a quadratic form, we signify these changes by writing $x_i'=P_i - x_i Q_i$ where $P_i, Q_i\in \pospoly$. Note that $P_i$ is a quadratic form and $Q_i$ is a linear form.

We can now perform time dilation with $x_0$. As a result, we have the cubic form system,
\[\begin{cases}
         x_i'= (P_i - Q_i x_i)\cdot x_0, \quad\text{for $i=1,\dots, n$}\\
         x_0'= -\sum_{i=1}^{n}x_i'
        \end{cases} \]
where we let $x_0(0)=1$ and $x_i(0)=0$ for $i=1,\dots,n$. Note that   $P_i x_0,Q_ix_0\in\pospoly$ and $x_i$ is in every negative term of $x_i'$ for $i=1,\dots,n$.
Some algebraic manipulation reveals that
$$x_0' = -\sum\limits_{i=1}^n x_i' =  -\sum\limits_{i=1}^n (P_i-Q_i x_i)\cdot x_0=  \sum\limits_{i=1}^n (Q_i x_i x_0- P_i x_0).$$Clearly, $Q_ix_0, P_ix_0\in \pospoly$ and $x_0$ is in every negative term of $x_0'$. Thus, $\bx'$ is CRN implementable.

Since $x_0$ is a factor of every term in the cubic form $\bx'$ , $x_i^3$ is not a positive term of $x_i'$ for $i=1,\dots,n$. For the case of $x_0'$, recall that the positive terms of $x_0'$ come from the negative terms of $x_1',\dots,x_n'$. Given that every negative term of $x_i'$ will have $x_i$ as a factor, every positive term of $x_0'$ will have at least one $x_i\neq x_0$ as a factor. Thus, $x_0^3$ cannot occur as a positive term in $x_0'$.

Hence, we have satisfied Corollary~\ref{cor:kPP} to obtain a TPP-implementable cubic form system.
\end{proof}

\begin{remark}
Note that in the above proof $x_0$ is bounded from below, away from zero. We can observe that:
\begin{enumerate}
    \item The variable $x_0$ starts from $1$ and remains greater than $c>0$ for all $t$. It will not introduce new equilibria.
    \item The new system has at most a linear slowdown compared to the old system, since $\int_0^t x(0) \diff t> \int_0^t c \diff t= ct$.
\end{enumerate}
\end{remark}
The variable $x_0$ acts like a source, distributing the total mass (sum $= 1$) to all other variables (initialized at 0).
\subsection{Stage 3: Conversion to PP-implementable Quadratic Form Systems}

\begin{theorem}\label{thm:pp-implementable-system}
    A number computable by a TPP-implementable function given by Stage 2 can be computed by the sum of a finite set of PP-implementable functions.
\end{theorem}
\begin{proof}(sketch)
    We utilize the self-product of the TPP-implementable cubic form system. Let $\bx= (x_0, x_1, \cdots, x_n)$ be the system resulting from the previous stage that computes a real number $\alpha$ with initial values $x_0 =1$ and 0 otherwise. We consider the system $\big(z_{i,j}\big)_{(i,j)\in [n+1]^2}$ where $z_{i,j}=x_i \cdot x_j$ and $[n+1]=\{0,1,\cdots,n\}$. Let $z_{0,0}(0)=1$ and 0 otherwise. The marked set $M_z= \{z_{1,j}| 0\leq j\leq n\}$. Then we can verify the $z$-system is PP-implementable and the sum of the marked variables traces $\alpha$. Check the appendix for a full proof.
\end{proof}
\begin{proof}(full)

Let $\bx=(x_0,x_1,\dots,x_n)$ be the TPP-implementable cubic form system from Stage  $2$. Then we have
\[\bx' = \begin{cases}
         x_i'= (P_i - Q_i x_i)\cdot x_0, \quad\text{for $i=1,\dots, n$}\\
         x_0'= -\sum_{i=1}^{n}x_i'
        \end{cases} \]
where  $x_1$ computes some $\alpha$ and the system is initialized such that  $x_0(0)=1$ and $0$ otherwise.

We construct our proof based on the ``one trick'' given by Observation~\ref{obs:one-trick},
\[
\sum_{i=0} ^n x_i  \cdot \sum_{j=0} ^n x_j =1\cdot 1 =1.
\]
We begin by defining a new variable vector, $\mathbf{z}$.
Introduce new variables $z_{i,j}=x_ix_j$  for all $(x_i,x_j)\in \bx\times \bx$. To compute $\alpha$, define the marked states as $M_z=\{z_{1,j} \mid 0\leq j\leq n \}$. Given $\bx(0)$, let $z_{0,0}(0)=1$ and $0$ otherwise.

Now, we find $\mathbf{z'}$. For all $ z_{i,j} \in \bz$ we have $z_{i,j}'=(x_ix_j)'=x_i'x_j + x_ix_j'$, where $x_i', x_j'\in \bx'$. Recall that $\bx'$ consists of cubic forms. Therefore, $x_i'x_j + x_ix_j'$ is a quartic form in variable $\bx$. It follows from a change of variables from $\bx$ to $\mathbf{z}$ that $z_{i,j}'$ is a quadratic form.
For each $z_{i,j}\in \mathbf{z}$, we check conditions $(i)$ and $(ii)$ of Theorem~\ref{thm:PP-cond}. There are three non-trivial cases.

Case 1: Assume $z_{i,j} = x_ix_j$ such that $i\neq 0$ and $j\neq 0$.
Then by the chain rule and substitution of $\bx'$, we have
\[ z_{i,j}' = x_i'x_j+x_ix_j' =  (x_0 P_i - x_i x_0 Q_i )\cdot x_j + x_i\cdot (x_0 P_j - x_j x_0 Q_j ).
 \]We then distribute and group terms by their sign to obtain:
 \[ z_{i,j}'= ( x_j x_0  P_i + x_i x_0 P_j) -( x_i x_j x_0 Q_i  + x_i  x_j x_0 Q_j )\]
 Note that we can factor $(x_i x_j)$ from every negative term,
 \[z_{i,j}'=  ( x_j x_0  P_i + x_i x_0 P_j) -( x_i x_j )(x_0 Q_i  + x_0 Q_j ).
\]
\begin{align*}
    z_{i,j}' &=  ( x_j x_0  P_i + x_i x_0 P_j) -( x_i x_j )(x_0 Q_i  + x_0 Q_j ) \\
            & = (z_{0,j} P_i + z_{0,i} P_j) - z_{i,j}\cdot (x_0 Q_i  + x_0 Q_j )
\end{align*}
and the terms $P_i$, $P_j$, $x_0Q_i$, and $x_0Q_j$ can be further written into linear form of $\{z_{i, j}\}_{i,j}$'s. We will not repeat these argument in the rest of the proof, though.

The above  implies $z_{i,j}$ is a factor of every negative term. Also, since $x_0$ can be factored from every positive term, we cannot have a positive $z_{i,j}^2$ (recall that $i, j$ are nonzero). Thus, conditions $(i)$ and $(ii)$ are satisfied.

Case 2: Assume $z_{0,j}$ such that $j\neq 0$.  From the chain rule and substituting equations from $\bx'$, we have
\[z_{0,j}' = x_0' x_j + x_0 x_j'=\sum\limits_{i=1}^n  ( x_i x_0 Q_i  -x_0 P_i ) \cdot x_j + x_0\cdot (x_0 P_j -x_j x_0 Q_j)\]
Distribute and group terms by their sign to obtain,
\[z_{0,j}' =  x_0^2 P_j +  \sum\limits_{i=1}^n  x_i x_j x_0 Q_i     - (x_j x_0^2 Q_j + \sum\limits_{i=1}^n  x_0 x_j P_i)\]
 Note that we can factor $(x_0 x_j)$ from every negative term,
\[z_{0,j}' =  x_0^2 P_j +  \sum\limits_{i=1}^n  x_i x_j x_0 Q_i     - (x_j x_0)(x_0 Q_j + \sum\limits_{i=1}^n P_i)\]

which implies $z_{0,j}$ is a factor of every negative term. Suppose a positive term in $z'_{0,j}$ has the monomial $(x_0x_j)^2$. By carefully choosing the change of variable (i.e., choose $z_{j,j}z_{0,0}$, not $z_{0,j}^2$), we satisfy condition $(ii)$. If there is no positive term in $z'_{0,j}$ with the monomial $(x_0x_j)^2$, then condition $(ii)$ is met. Then, both conditions $(i)$ and $(ii)$ are met.


Case 3: Assume $z_{0,0}$.
With the chain rule and substitution, we have
\[ z_{0,0}' = 2x_0 \cdot x_0' = 2x_0 \cdot\sum\limits_{i=1}^ n (x_0 x_i Q_i -x_0 P_i)\]
Distributing and grouping terms by their sign yields,
\[ z_{0,0}' = \sum\limits_{i=1}^n (2 x_0^2 x_i Q_i) -\sum\limits_{i=1}^n (2 x_0^2 P_i)\]
Since $x_0^2$ can be factored from every negative term, condition $(i)$ is satisfied. Since   $i=1,\dots,n$, it follows that $x_i\neq x_0$. Thus, condition $(ii)$ is satisfied.

In summary, we have defined the states $\mathbf{z}$, marked states $M_\mathbf{z}$, initial configuration $\mathbf{z}(0)$, and differential system $\mathbf{z'}$. Furthermore, for each $z_{i,j}'\in \mathbf{z'}$ we have established: each $z_{i,j}'$ is a quadratic form, $z_{i,j}'$ can be written to avoid a positive $z_{i,j}^2$ in $z_{i,j}'$, and the existence of $z_{i,j}$ in every negative term of $z_{i,j}'$. Hence, a TPP-implementable cubic form system can be transformed into a PP-implementable system that computes $\alpha$.

\end{proof}

\subsection{Stage 4: Converting PP-implementable Quadratic Form Systems to PLPPs}
In this subsection, we turn the PP-implementable quadratic form given as a result of the previous stage into a PLPP. The following construction is from the proof of Lemma 5 in \cite{cBoFrKo12}.
\begin{construction}\label{construction:PLPP}
Let $\bxn$ be a PP-implementable quadratic form. We denote the derivative $x_i$ as
\begin{equation}\label{eq:x_i_prime_g_form}
    x_i' = p_i-q_i x_i, \quad{i=1,2,\cdots,n,}
\end{equation}
where $p_i$ is a quadratic form and $q_i$ is a linear form.

We will apply the $\varepsilon$-trick (Operation~\ref{op:ep-trick}) to shrink the coefficient of the above system, where $\varepsilon \in \Q$ is an undetermined parameter. In order to align with the form in Equation~(\ref{eq:component_balance}), we add and subtract a term $2x_i$. Now the system has the form
\begin{align}\label{eq:plus_2x_minus_2x}
    x_i' &= \varepsilon (p_i-q_i x_i)+ 2x_i -2x_i, \quad{i=1,2,\cdots,n.}\\
         &= f_i(\bx), \quad\text{denote $f_i(\bx)= \varepsilon (p_i-q_i x_i)+2x_i$.}
\end{align}
We then use the one-trick to rewrite $f_i(\bx)= \varepsilon (p_i-q_i x_i)+2x_i\sum_{j=1}^n x_j$. Note that the term brings in $2x_ix_j$ for $j=1, \cdots, n$. We will then pick $\varepsilon$ sufficiently small
    \begin{enumerate}
        \item to ensure the $2x_i x_j$ terms dominate the monomials in $-qx_i$ such that $f_i(\bx)\in\pospoly$ for all $i$.
        \item and to ensure
     \begin{equation}
         \frac{\sum_{i, j\not=k}\coeff(x_jx_k, f_i)}{4}\leq 1 \quad\text{and}\quad
              \frac{\sum_{i, j}\coeff(x_{j}^2, f_i)}{2}\leq 1
     \end{equation}
     where $\coeff(x_jx_k, f_i)$ denotes the coefficient of $x_jx_k$ in $f_i$, and $i,j,k\in \{1,\cdots,n\}$.
    \end{enumerate}

Next, we construct a rule set for a PLPP. For an ordered pair $(x_j, x_k)$ and a term $\alpha_{i,j,k} x_j x_k$ in $f_i(\bx)$ we do the following:
      \begin{itemize}
          \item If $j=k$, add a rule $(x_j,x_j)\to \frac{\alpha_{i,j,k}}{2} (x_i,x_i);$
          \item otherwise, add two rules $(x_j,x_k)\to \frac{\alpha_{i,j,k}}{4} (x_i,x_i)$ and $(x_k, x_j)\to \frac{\alpha_{i,j,k}}{4} (x_i, x_i);$
          \item In the above rules, if $s=\sum_{i=1}^n \alpha_{i,j,k}< 1$ for an order pair $(k,l)$, we add an idle reaction $$(x_j, x_k)\to (1-s) (x_j, x_k);$$
      \end{itemize}
If a term $x_j x_k$ appears in $f_i(\bx)$, it implies that the pair $(x_j, x_k)$ will produce $x_i$. We use the all-in greedy strategy, that is, we produce two $x_i$'s and let the pair $(x_k,x_j)$ do the same. So we need to divide the coefficient $\alpha_{i,j,k}$ by 4 for $j\not=k$; yet, when $k=j$ we don't have a flipped order pair, so we only divide the coefficient by 2.
\end{construction}

There need to make sure that the rule set generated above forms a PLPP.
\begin{theorem}\label{thm:PLPP}
    Let $\bxn$ be a quadratic system satisfying $\sum_{i=1}^n x_i =1$.
    Then it can be turn into a PLPP by Construction~\ref{construction:PLPP} if and only if it is PP-implementable.
\end{theorem}
\begin{proof}
    The ``only if'' part is clear. Therefore, it suffices to show the ``if'' direction.
    We adopt the same symbols and notations as in Construction~\ref{construction:PLPP}. Recall that we denote $f_i(\bx)= \varepsilon (p_i-q_i x_i)+2x_i\sum_{j=1}^n x_j$. Consider an arbitrary ordered pair $(x_j, x_k)$. We need to show that in the construction:
    \begin{enumerate}
        \item  The selection of $\varepsilon$ satisfying the two conditions is always possible.
        \item  The rule set generated by the construction forms an PLPP. Specifically, for the ordered pair $(x_j, x_k)$, we must show that $\sum_{i=1}^n \alpha_{i, j,k}$ =1 and $\alpha_{i,j,k}>0$ for all $i$. That is, the $\alpha_{i,j,k}$'s form a discrete probability measure with $j,k$ fixed.
    \end{enumerate}
    We break the proof into two cases:
    \begin{itemize}
        \item Case $j\not = k$:
            The term $x_j x_k$ occurs in $f_i(\bx)$ from two sources.
    \begin{itemize}
        \item Case 1: $ x_j x_k \in \varepsilon (p_i-q_i x_i)$. The coefficients of such terms for all $i$ must sum to zero, since the system is assumed to be PP-implemetable.
        \item Case 2:  $x_j x_k  \in 2x_i(x_1+\cdots+ x_n)$. We have two sources:
        $x_j x_k$ occurs in $f_j(\bx)$ as a term in $2x_j(x_1+\cdots x_k + \cdots+ x_n)$ and similarly in $f_k(\bx)$ as a term in $2x_k(x_1+\cdots x_j + \cdots+ x_n)$.
    \end{itemize}
      The sum of these coefficients is $4$. Note that the coefficients are divided by 4 in the construction, which ensures the resulting sum of $\alpha_{i,j,k}$ is one. One can verify that the selection of $\varepsilon$ is always possible under this case.
    \end{itemize}
    \item Case $j= k$: For term $x_j^2$, which occurs in $f_i(\bx)$, there are two subcases:
        \begin{itemize}
            \item $i\not=j$: Then $x_j^2 \in \varepsilon p_i$ and the coefficient must be positive.
            \item $i=j$: Then $x_j^2=x_i^2 \in -\varepsilon q + 2x_i \sum_{k=0}^n x_k$ and coefficient is less than 2. ($x_j^2=x_i^2$ can not occur in $\varepsilon p$ since the system is PP-implementable.)
        \end{itemize}
        Recall that when constructing the rule set, the coefficients are divided by 2. As a result, generated $\alpha_{i,j}$'s are positive and less than 1. Due to the system's conservative nature, the coefficients $x_j^2$ in $\varepsilon (p_i -q_i x_i)$ for $i$'s in $\{1,\cdots, n\}$ will sum to zero.  Since the  $2x_j \sum_{k=1}^n x_k$ term introduces a $2x_j^2$ term,  the resulting coefficient sum for $x_j^2$ is $2$. In the construction, these coefficients are always divided by 2 to generate $\alpha_{i,j}$ in $f_i(\bx)$. Thus, the $\sum_{i=1}^n \alpha_{i,j} =1$. One can then verify that the selection of $\varepsilon$ is feasible.

        In contrast, assume the system is not PP-implementable. Then there is an $x_i^2 \in \varepsilon p_i$ with a positive coefficient. Once combined with the term $2x_i^2 \in 2x_i \sum_{k=1}^{n} x_k$, the sum of coefficients will be greater than two. Then dividing by 2 will generate an $\alpha_{i, i}>1$, which is not permitted by PLPPs. Hence, the selection of $\varepsilon$ is not feasible in this case.
\end{proof}

The following theorem can help in transforming a PLPP to an LPP.
\begin{theorem}[\cite{cBoFrKo12}, Lemma 4]\label{thm:derandomization}
Let $\nu\in[0, 1]$ , and assume that there exists a  PLPP computing $\nu$, with rational probabilities. Then there exists a (deterministic) LPP computing $\nu$.
\end{theorem}
Note that although the above theorem is based on the original notion of LPP-computable number with the finite-equilibria constraint, the construction in the proof of the theorem still correctly encodes the balance equation. So the theorem also holds under the new definition.

We therefore finish the proof of our main theorem.
\begin{main-theorem}\label{thm:main}
    LPPs compute the same set of numbers in [0,1] as GPACs and CRNs.
\end{main-theorem}
\begin{corollary}
    Algebraic numbers are LPP computable.
\end{corollary}
\begin{proof}
    Algebraic numbers are computable computable by a GPAC \cite{cBoFrKo12} (Lemma 5) so they are also LPP-computable by the main theorem.
\end{proof}
The above corollary fix the gap in Bournez et al.’s construction.

\begin{corollary}\label{thm:trans}
    Famous math constants $\frac{\pi}{4}$, $e^{-1}$, Euler's $\gamma$, Dottie number, and Catalan's constant are LPP-computable.
\end{corollary}
\begin{proof}
    The first four numbers are GPAC/CRN computable (see \cite{huang2020chemical}). Note that all real-time computable numbers by CRNs in \cite{huang2019real} are computable in this paper. One just need to disregard the real-time constraint. Then by our main theorem, they are also LPP computable. It suffices to show that Catalan's constant is GPAC computable. We use the formula
$G = \frac{1}{2}\int_0^\infty \frac{t}{\cosh{t}} dt.$
We have
    \begin{align*} G&=\int_0^\infty \frac{t}{e^t + e^{-t}}dt,\quad\text{since $cosh(t)=\frac{e^{t}+e^{-t}}{2}$}\\
&=\int_0^\infty (t\cdot e^{-t})\cdot \frac{1}{1+e^{-2t}}dt.
\end{align*}

We let
\[
    R= t\cdot e^{-t}, \quad E=e^{-t}, \quad V = \frac{e^{-2t}}{1+ e^{-2t}},
\]
and
\[
    G(t) =\int_0^t  (te^{-t})\cdot \frac{1}{1+e^{-2t}}dt = \int_0^t  R\cdot(1-V)dt.
\]
Taking derivatives against $t$, we get
\[
    \begin{cases}
        G^{\, \prime} = R(1-V)\\
        R^{\, \prime} = E - R \\
        E^{\, \prime}= -E \\
        V^{\, \prime}= (1-V)^2 \cdot (-2E^2),
    \end{cases}
    \quad
        \begin{cases}
        G(0) = 0\\
        R(0) = 0 \\
        E(0)=1 \\
        V(0)=\frac{1}{2}.
    \end{cases}
\]
Therefore Catalan's constant can be computed by the above PIVP. Note that the system's non-zero initial values can be transformed into a system with zeroed initial values (\cite{huang2019real}, Theorem 3). Hence by our Main Theorem, Catalan's constant is LPP computable.
\end{proof}

This result answers Bournez et al.'s question about whether solutions of trigonometric equations (Dottie number) and $\frac{\pi}{4}$ are ``computable'' asymptotically by population protocols.

\section{Conclusion and Discussion}\label{sec:conclusion}
In this paper, we extend the notion of LLP-computable numbers and connect it with GPAC/CRN-computable numbers. Moreover, we show that LLPs and GPACs/CRNs essentially compute the same set of real numbers. The result, to some degree, says LPPs, or more straightforwardly, population protocols with the mass-action semantic, can simulate GPACs in some way. We would ask a natural question: What is the next ``weak'' or minimal (according to some measure) model that can simulate GPACs? Unimolecular protocols are provably incapable since they only compute rational numbers. However, a model discussed in \cite{gerin2012algebraic}, which can be viewed as a two-state (``black'' or ``white'') $k$-PP with restrictions on reactions (the product of a reaction must either be all black or all white), can compute algebraic numbers such as $\frac{3-\sqrt{5}}{2}$ but not rational numbers like $\frac{1}{5}$. Diving into the sub-models of population protocols and characterizing the power spectrum will be an interesting avenue for further exploration.

The ODE systems we used to compute transcendental numbers have a continuum of equilibria. We have not addressed the semistability nor the system's convergence rate in this paper. Regarding the latter, since our construction essentially rewrites the system with a new set of auxiliary variables, the system's solutions in some sense remain the same, except for the \emph{linear} slowdown introduced in Stage 2. It is hoped that an LPP's convergence rate is similar to the input CRN. If an input CRN converges exponentially fast to a number $\alpha$, we desire the translated LPP also converges (exponentially) fast. We would need a thorough stochastic analysis along this line in future work.

The definition of GPAC/CRN-computable numbers and LPP-computable numbers have a major difference: We can trace a real number by a \emph{set} of variables in an LPP, rather than one, as in a GPAC/CRN. Currently, the proof of our main theorem relies on using a set of marked variables. Our question: Is the difference intrinsic? Or can we compute the same set of numbers using \emph{one} variable?

Another area of exploration is protocols with a large population (tending to infinity) that also have ``scarce'' variables with a limited population. The behaviors of these systems, due to the existence of the scarce variables, can not be governed by Kurtz's theorem. Therefore, a new analytic tool or theoretical benchmark (probably not a computable number) awaits development.

Our construction in Section~\ref{sec:e_of_lpp_and_gpacs} is not optimal. To implement the translation algorithm in a software package, one would want to have a better algorithm for Stage 1. To apply the $\lambda$-trick, one needs to know the bounds of the variables, which would require numerical methods. A theory to predict the bounds is desired but difficult to achieve due to the non-linear nature of ODEs.

\section*{Acknowledgment}
We thank Titus Klinge for providing the motivating example in Section 1.

We thank the anonymous reviewers for their careful reading of our manuscript and their many
insightful comments and suggestions.


\bibliographystyle{plain}
\bibliography{master}
\end{document}